\documentclass{article}

\usepackage[preprint]{neurips_2019}

\usepackage[utf8]{inputenc} % allow utf-8 input
\usepackage[T1]{fontenc}    % use 8-bit T1 fonts
\usepackage{hyperref}       % hyperlinks
\usepackage{url}            % simple URL typesetting
\usepackage{booktabs}       % professional-quality tables
\usepackage{amsfonts}       % blackboard math symbols
\usepackage{nicefrac}       % compact symbols for 1/2, etc.
\usepackage{microtype}      % microtypography
\pdfoutput=1

\usepackage{url}
\usepackage{amsmath}
 
\usepackage{amsthm}
\usepackage{algorithm}
\usepackage[noend]{algorithmic}
\usepackage{verbatim}
\usepackage{graphicx}
\usepackage[space]{grffile}
\usepackage{subfigure}
\usepackage{tikz}
\usetikzlibrary{arrows}
\newtheorem{theorem}{Theorem}

\newtheorem{lemma}{Lemma}

\newtheorem{corollary}{Corollary}
\newtheorem{proposition}{Proposition}

\theoremstyle{definition}

\newtheorem*{deff}{Definition}
\newtheorem{claim}{Claim}
\newtheorem{examp}{Example}

\graphicspath{{figures/}}
\usepackage{amsmath}
\usepackage{mathrsfs}
\usepackage{amssymb}
\usepackage{booktabs}
\usepackage{changepage}
\usepackage{xspace}
\usepackage{hyperref}

\newcommand{\eg}{\textit{e.g.}\xspace}

\newcommand{\epsi}[0]{ \varepsilon }
\newcommand{\reals}[0]{ \mathbb{R}_{\ge0} }

\newcommand{\ff}[1]{ f \left( #1 \right) }

\newcommand{\tg}{\texttt{TG}\xspace}
\newcommand{\threshold}{\texttt{THRESHOLD}\xspace}

\newcommand{\maxuni}{\texttt{MAX-UNION}\xspace}

\newcommand{\isg}{\texttt{ISG}\xspace}
\newcommand{\maxiOne}{\ensuremath{\texttt{MAXI}_1}\xspace}
\newcommand{\maxiP}{\ensuremath{\texttt{MAXI}_p}\xspace}
\newcommand{\maxi}{\ensuremath{\texttt{MAXI}}\xspace}
\newcommand{\unc}{\texttt{UNCONSTRAINED-MAX}\xspace}

\newcommand{\stopGain}{\epsi M/ n}

\newcommand{\psystem}{\mathcal I}
\newcommand{\pint}{\mathcal I}

\DeclareMathOperator*{\argmax}{arg\,max}

\usepackage{tikz}
\usetikzlibrary{arrows}

\usepackage{xcolor}

\hypersetup{
     colorlinks,
     linkcolor={red!80!black},
     citecolor={blue!80!black},
     urlcolor={blue!80!black}
}

\title{A Note on Submodular Maximization over Independence Systems}

\author{%
  Alan Kuhnle\thanks{Webpage: \url{http://www.alankuhnle.com}} \\
  Department of Computer Science \\
  Florida State University\\
  Tallahassee, FL \\
}

\begin{document}

\maketitle

\begin{abstract}
  In this work, we consider the maximization of submodular functions constrained
  by independence systems. Because of the wide applicability
  of submodular functions, this problem has been extensively studied in
  the literature, on specialized independence systems. 
  For general independence systems, even when 
  all of the bases of the independence system
  have the same size,
  we show that for any $\epsilon > 0$,
  the problem is hard to approximate within $(2/n)^{1-\epsilon}$, where 
  $n$ is the size of the ground set. In the same context,
  we show the greedy algorithm does obtain a ratio of $2/n$ under
  a mild additional assumption.
  Finally,
  we provide the first nearly linear-time algorithm for 
  maximization of non-monotone
  submodular functions over $p$-extendible independence systems.
\end{abstract}

\section{Introduction} \label{sect:intro}
\begin{algorithm}
   \caption{\texttt{GREEDY}$(f, \psystem )$: The Greedy Algorithm}
   \label{alg:greedy}
   \begin{algorithmic}[1]
     \STATE {\bfseries Input:} $f : 2^U \to \reals$, $\psystem$: independence system
     \STATE {\bfseries Output:} $G \subseteq U$, such that $G \in \pint$.
     \WHILE {$G$ is not maximal in $\pint$}
     \STATE $g \gets \argmax_{s \in U : G \cup \{ s \} \in \pint} f( G \cup \{ s \} )$
     \STATE $G \gets G \cup \{ g \}$
     \ENDWHILE
     \STATE \textbf{return} $G$
\end{algorithmic}
\end{algorithm}
Submodularity\footnote{A function $f:2^U \to \reals$ is \emph{submodular} if for every 
$S \subseteq T \subseteq U$, $x \in U \setminus T$, 
$f( T \cup \{ x \} ) - f(T) \le f( S \cup \{ x \} ) - f( S )$.} captures an important diminishing-returns property of 
discrete functions. 
Submodular set functions arise from
\eg
viral marketing \citep{Kempe2003}, data summarization \citep{Mirzasoleiman2015},
and sensor placement \citep{Krause2008a}.
The optimization
of these functions has been studied subject to various types of
independence system\footnote{An \emph{independence system} $\mathcal I$ 
on the set $U$ is a collection of subsets of $U$
such that 
(i) $\mathcal I$ is nonempty, and
(ii) if $S \in \mathcal I$ and
  $T \subseteq S$, then $T \in \mathcal I$.} constraints, including
cardinality \citep{Nemhauser1978}, matroid \citep{Fisher1978}, and the more general independence
systems \citep{Calinescu2011}. 
Formally, 
the problem (\texttt{MAXI}) considered
in this work is the following: given submodular function
$f:2^U \to \reals$ and independence
system $\mathcal I$ on $U$, determine
\begin{equation*}
  \argmax_{S \in \mathcal I} f(S).
\end{equation*}
% The subproblem of \maxi wherein $\mathcal I$ is restricted an independence system
% in which the bases, or maximal independent sets, differ in size by up to a 
% factor of $p$,
% is denoted \maxiP.

% Because of its ubiquity, 
% the problem \maxiP has been extensively studied:
% \citet{Calinescu2011} generalized an argument
% of \citet{Jenkyns1976} from modular to submodular functions to show that
% the greedy algorithm (Alg. \ref{alg:greedy})
% achieves a ratio of $1/(p+1)$ on \maxiP; similar arguments appear
% in the proofs for fast algorithms on $p$-systems \citep{Badanidiyuru2014}
% and for the maximization of non-monotone submodular functions \citep{Gupta2010}.
% Since \citet{Badanidiyuru2014} showed no ratio 
% better than $1/(p-\epsilon)$ is possible within polynomially many
% queries to $f$,
% it has been believed that a greedy approach is nearly optimal
% on $p$-systems.
% In this work, we show that the approximation
% ratios for the greedy algorithm on \maxiP 
% are incorrect; the difficulty of optimizing
% submodular functions on $p$-systems has been underestimated. 
Even on an independence system where
maximal independent sets have the same size, the greedy algorithm
may return arbitrarily bad solutions for \maxi.
% This result is perhaps surprising: when the independence system is
% a matroid, which is a $1$-system with an exchange property between
% independent sets, 
% the greedy algorithm gives a $1/2$-approximation. Indeed
Our results
indicate that 
some exchange property between independent sets must exist if
the problem is to be tractable. 
\paragraph{Contributions}
Our main contributions are summarized as follows. 
\begin{itemize}
\item Let \maxiOne denote the 
  subclass of independences systems where maximal independent sets have the same size.
  We show that \maxiOne admits no polynomial-time algorithm
  with approximation ratio better than $(2/n)^{1-\epsilon}$ unless NP = ZPP,
  even when the submodular function $f$ is restricted to be monotone; here,
  $n = |U|$ is the size of the ground set, and $\epsilon > 0$ is arbitrary.
% Intuitively, $p$-systems are difficult because of the lack
% of any sort of exchange property between independent sets. 
% This contradicts
% the paper of \citet{Jenkyns1976}, the arguments of which have
% been recently generalized to submodular functions:
On the other hand, under the condition that the system has two disjoint bases,
the greedy algorithm does obtain a ratio of $2/n$.
% that %a result by \citet{Calinescu2011} 
% the greedy algorithm has a ratio of $1/(p+1)$ for \maxiP with monotone objective 
% \citep{Calinescu2011,Badanidiyuru2014}. 
% An example is given in Section \ref{sect:greedy-bad} showing that 
% the greedy algorithm may perform arbitrarily badly on a general $p$-system;
Intuitively, the difficulty of approximation on a $p$-system arises from
the lack of any exchange property between the independent sets.

%ratio for the greedy algorithm is given in this paper. %Both \citet{Calinescu2011} and \citet{Badanidiyuru2014}
%use a similar argument to \citet{Jenkyns1976}.
%However, most results for monotone submodular functions 
%(e.g. \citep{}) extending \citet{} do hold for $p$-extendible
%systems. 

% \item For submodular functions that are not necessarily monotone,
% the hardness result stated above
% %a counterexample is provided
% also contradicts prior literature. %the results of \citet{Gupta2010} for
% %\maxiP.
% %below, a counterexample  is given to a crucial lemma of \citet{Gupta2010} 
% %for the greedy algorithm and $p$-systems.
% A key lemma of \citet{Gupta2010} for $p$-systems is invalidated;
% this lemma underlies the results of many recent works, 
% including \citet{Mirzasoleiman2016,Feldman2017,Mirzasoleiman2018}.
% We provide
% a replacement lemma which fixes \citet{Gupta2010}
% and all downstream works for \maxiP when the $p$-system is restricted
% to be a $p$-extendible system\footnote{$p$-Extendible systems form a subclass of $p$-systems. See Section \ref{sec:prelim} for the definition.}.
\item Also, we provide a deterministic algorithm TripleGreedy (Alg. \ref{alg:sg}),
which has the ratio $\approx 1/(4 + 2p)$ on $p$-extendible systems in $O( n \log n )$ 
function evaluations, when the objective function is submodular but not necessarily monotone. 
This is the first 
approximation algorithm on $p$-extendible systems whose runtime is 
linear up to a logarithmic factor in the size $n$ of the ground set and is 
independent of both $p$ and the the maximum size $k$ of any 
independent set.
In prior literature, the fastest randomized algorithm is that of \citet{Feldman2017}, which
achieves expected ratio $1 / (p + 2 + 1/p)$ in $O(n + nk/p)$ evaluations,
while the fastest deterministic algorithm is also by \citet{Feldman2017}
and achieves ratio $1/\left( p + O(\sqrt{p}) \right)$ in $O \left( nk \sqrt{p} \right)$
evaluations.

\end{itemize}
% Unlike the IteratedGreedy algorithm of \citet{Gupta2010}, 
% which requires $p$ executions of the
% standard greedy algorithm and time , TripleGreedy requires only two.

\paragraph{Related work}
The maximization of monotone, submodular functions over independence systems has a long
history of study; 
\citet{Fisher1978} proved the approximation ratio of $1/(p+1)$ for the greedy algorithm when
the independence system is an intersection of $p$ matroid constraints, which
is a special case of a $p$-extendible system. This ratio for the greedy algorithm
was extended to $p$-extendible systems by \citet{Calinescu2011}, as well
as to the more general $p$-system constraint. 
A similar ratio for a faster, thresholded greedy algorithm and $p$-system constraint 
was also given by \citet{Badanidiyuru2014}. 

For the special case when the independence system is a single matroid or cardinality
constraint, better approximation guarantess have been obtained:
in \citet{Calinescu2011}, an optimal $(1-1/e)$-approximation is given when $f$ is monotone
and the independence system is a matroid. For further information,
the reader is referred to the survey of \citet{Buchbinder2018a} and references therein.

When $f$ is non-monotone and the independence system is a $p$-extendible system, \citet{Gupta2010} provided an $\approx 1 / (3p)$-approximation in $O( nkp )$ function evaluations; this was improved by 
\citet{Mirzasoleiman2016} to $\approx 1 / (2p)$ with the same time complexity, and
\citet{Feldman2017} improved this to a ratio of $1/\left( p + O(\sqrt{p}) \right)$ in $O \left( nk \sqrt{p} \right)$ evaluations. Furthermore, \citet{Mirzasoleiman2018} extended these works to a streaming setting. All of these works rely upon an iterated greedy approach, which employs up to $p$ iterations of the standard greedy algorithm. 
%All results discussed in this paragraph were incorrectly shown for general $p$-systems, but apply to $p$-extendible systems instead by our replacement lemma in Section \ref{sect:gupta}.
In Section \ref{sec:triple}, we propose a simpler iterated greedy approach for $p$-extendible systems, which relies upon only two iterations of the greedy algorithm.
We show how to speed up this algorithm to obtain ratio $\approx 1/(2p)$ in $O(n \log n)$ evaluations.

\paragraph{Organization}
The rest of this paper is organized as follows: in Section \ref{sec:prelim} we define notions
used throughout the paper. In Section \ref{sect:inapprox} we prove the hardness 
result for \maxiOne. 
Next, we show that the greedy algorithm is indeed the optimal approximation on \maxiOne under
 a weak assumption in Section \ref{sect:greedy-ratio}. 
 Finally, in Section \ref{sec:triple}
we provide our nearly linear-time for submodular maximization over a $p$-extendible system. 
% \section{Preliminaries} \label{sec:prelim}
% $f|_A$ -- the function $f$ restricted to set $A$.
% $[l] = \{0,1,\ldots,l-1\}$

\section{Preliminaries} \label{sec:prelim}
Throughout the paper, $U$ denotes the ground set of size $n$.
In this work, the objective function is a non-negative function
$f:2^U \to \reals$; typically, the function $f$ is given
as an oracle that returns, for given set $A \subseteq U$,
the value $f(A)$. Our inapproximability result in Section
\ref{sect:inapprox} holds in this model, but it also
holds when a description of $f$ as a polynomial-time
computable function is given as input. When $A$ is a set and $x \in U$, we occasionally write $A + x$ for $A \cup \{ x \}$.

The members of an independence system 
are termed \emph{independent sets}. An independent set $A$ is a \emph{basis} of independence system $\mathcal I$ if for all $x \in U \setminus A$, $A \cup \{ x \} \not \in \mathcal I$.
\begin{deff}[Matroid]
An independence system $\mathcal I$ is a matroid if the following property 
holds: if $S_1,S_2 \in \mathcal I$ and $|S_1| > |S_2|$, then 
there exists $x \in S_2 \setminus S_1$ such that $S_1 \cup \{ x \} \in \mathcal I$.
\end{deff}
\begin{deff}[$p$-Extendible System]
An independence system $(U, \mathcal I)$ is $p$-extendible
if the following property holds.
If $A \in \mathcal I$, $B \in \mathcal I$ with $A \subsetneq B$ and
if $x \notin A$ such that $A \cup \{ x \} \in \mathcal I$, then there exists
subset $Y \subseteq B \setminus A$ with $|Y| \le p$ such that $B \setminus Y \cup \{ x \} \in \mathcal I$.
\end{deff}
\begin{deff}[$p$-System]
A $p$-system is an independence system $\mathcal I$ such that if $S_1,S_2 \in \mathcal I$
are bases, then $|S_1| / |S_2| \le p$. 
\end{deff}

We remark that every $p$-extendible system is also a $p$-system, but that the converse is not true,
as the exchange property defining a $p$-extendible system may not hold. Furthermore, every matroid is a $1$-system, but the converse does not hold. As an example, let $n = 4$, $U = \{a,b,c,d\}$, and $\mathcal J = \{ \emptyset, \{ a \}, \{ b \}, \{ c \}, \{ d \}, \{ a, b \}, \{ c, d \} \}$. Then $\mathcal J$ is clearly a $1$-system but not a matroid.

\section{Hardness of Submodular Maximization over Independence Systems} \label{sect:inapprox}

% \begin{figure}
%   \subfigure[test] {
%     \resizebox{0.4\textwidth}{!}{
%     \begin{tikzpicture}
%       \begin{scope}[every node/.style={circle,thick,draw}]
%         \node (0) at (0,1) {$s_1$};
%         \node (1) at (2,1) {$v_1$};
%         \node (2) at (-2,1) {$v_2$};
%         \node (3) at (-2,0) {$s_4$};
%         \node (4) at (-1,0) {$v_3$};
%         \node (5) at (1,0) {$v_4$};
%         \node (6) at (2,0) {$t_5$};
%         \node (7) at (3,0) {$s_3$};
%         \node (8) at (-3,0) {$s_2$};
%         \node (9) at (-2,-1) {$t_1$};
%         \node (10) at (2,-1) {$t_2$};
%         \node (11) at (3,-1) {$t_4$};
%         \node (12) at (-3,-1) {$t_3$};
%       \end{scope}
%       \begin{scope}[every node/.style={fill=white,circle},
%         every edge/.style={draw=black,very thick}]
%         \path [->] (0) edge (1);
%         \path [<-] (0) edge (2);
%         \path [->] (1) edge (6);
%         \path [<-] (2) edge (3);
%         \path [<-] (3) edge (9);
%         \path [->] (6) edge (10);
%         \path [<-] (8) edge (3);
%         \path [->] (12) edge (9);
%         \path [<-] (6) edge (7);
%         \path [->] (10) edge (11);
%         \path [->] (3) edge (4);
%         \path [->] (4) edge (5);
%         \path [->] (5) edge (6);
%       \end{scope}
%     \end{tikzpicture} } }
% %\caption{$\mathcal S = 
% %  \{(t_1,s_1),(s_1,t_2),(t_3,s_2),(s_3,t_4),(s_4,t_5)\}$; 
% %  Every edge of $G$ lies upon a shortest path of length 3 
% %  between one of these pairs.} \label{fig:not-lower-bd}
% \end{figure}
In this section, the main inapproximability result is proven for \maxiOne: 
maximization of submodular functions over independence systems for which
all maximal bases have equal size.

Hardness of \maxiOne is established via an approximation-preserving reduction to the
independent set problem (\isg) in a graph, 
which is to find the maximum size of an edge-independent set of vertices.
Once this reduction is defined, we show that any $\alpha$-approximation
for \maxiOne yields an $\alpha$-approximation for \isg, and our hardness
result follows from the hardness of \isg.
%for any $\epsilon > 0$,
%there is no polynomial-time algorithm to approximate \maxiOne within $(2/n)^{1-\epsilon}$, where 
%$n$ is the size of the ground set $U$, unless NP = ZPP.

\begin{deff}[\isg]
The \isg problem is the following: given a finite graph $G=(V,E)$, where
$E \subseteq V \times V$, define a set
$A \subseteq V$ to be edge-independent iff no pair of vertices in $A$ 
have an edge between them. Then the \isg problem is to determine the maximum
size of an edge-independent set in $V$.
\end{deff}
It is easily seen that the set 
$\mathcal I_G = \{ V : V$ is edge-independent in $G\}$ is an independence
system. In general, $\mathcal I_G$ may be a $(m - 1)$-system, where $m = |V|$; 
consider a star graph where all vertices are connected to a center vertex and no
other edges exist. %\citet{Hastad1999} showed that for any $\epsilon > 0$,
%it is hard to approximate \isg within factor $(1/m)^{1 - \epsilon}$.

Intuitively, the reduction works by transforming
a graph, which is an instance of \isg, into an instance
of \maxiOne through the padding of edge-independent sets
with dummy elements so that maximal independent sets have the same size.
A submodular function is then defined that maps the padded independent sets
to the size of the original, unpadded, edge-independent set in the graph. 
Formally, the reduction is defined
as follows.
\begin{deff}[Reduction $\Phi$]
  Let $G = (V,E)$ be a graph, which is an instance of \isg.
  Let $U = V \dot{\cup} D$, where $D$ is a set of $n = |V|$ dummy elements.
  An independence system $\mathcal I$ is defined on $U$ as follows:
  $S \subseteq U$ is in $\mathcal I$ iff. $S \cap V$ is
  edge-independent in $G$ and $|S \cap D| \le n - |S \cap V|$.
  Define function $f:2^U \to \reals$,
  by $f( S ) = |S \cap V|$. 
\end{deff}
We remark that the function $f$ is defined on all subsets of $U = V \cup D$,
not only members of the independence system. 
To illustrate the reduction, we provide the following example.
\begin{examp}
  Let $G = (V,E)$ be a star graph with five vertices. That is,
  $V = \{s, a, b, c, d\}$ and $E = \{ (s,a),(s,b),(s,c),(s,d) \}$.
  Then the maximal, edge-independent sets are $\{ s \}$ and $\{a,b,c,d\}$. 
  Then $\Phi$ maps this graph to the following independence system.
  The ground set $U = \{s,a,b,c,d\} \cup D$,
  where $D$ is a set of five dummy elements. Then the independence system
  $\mathcal I$ defined by $\Phi$ has bases $$\mathcal B = \left\{\{ a,b,c,d, e \}: e \in D \right\} \cup \left\{ \{ s, e_1, e_2, e_3, e_4 \} : e_i \in D, 1 \le i \le 4 \right\}. $$
  That is, $\mathcal I$ consists of all subsets of elements of $\mathcal B$.
\end{examp}
By the following lemma, the reduction $\Phi$ takes an instance of
\isg to an instance $(\mathcal I, f)$ of $\texttt{MAXI}_1$. 
Notice that the independence of any subset $B$ of $U$ may be
checked in polynomial time; the same is true for computation of $f(B)$.
\begin{lemma}
  Let $G$ be an instance of \isg, and let $\Phi (G) = (\mathcal I, f)$. Then
  \begin{itemize}
    \item[(i)] $\mathcal I$ is an independence system; in particular, all maximal
      bases have equal size.
    \item[(ii)] $f$ is monotone and submodular. 
      
  \end{itemize}
\end{lemma}
\begin{proof}
  %\paragraph{(i)}
  \textbf{(i)}:
    Clearly, $\mathcal I$ is non-empty, since any singleton vertex $v$
    is edge-independent in $G$, and $\{ v \} \in \mathcal I$. Furthermore, it is closed under
    subsets: let $S = A \dot{\cup} B \in \mathcal I$, where $A \subseteq V$, 
    $B \subseteq D$, and let $T \subseteq S$. Then $T = \hat{A} \dot{\cup} \hat{B}$,
    where $\hat{A} \subseteq A$, $\hat{B} \subseteq B$. 
    Since any subset of an edge-independent set of $G$ is also
    edge-independent, we have that $\hat{A}$ is edge-independent in $G$, and 
    $$|T \cap D | = |\hat{B}| \le |B| \le n - |A| \le n - |\hat{A}| = n - |T \cap V|.$$
    Hence $T \in \mathcal I$. Thus, $\mathcal I$ is an independence system on $U$.

    Next, suppose $S = A \dot{\cup} B \in \mathcal I$ is maximal. 
    Then $|S| = |A| + |B| = n$, for otherwise another dummy element
    could be added to $B$ to produce a larger independent set.
    Hence $\mathcal I$ is a $1$-system.

    \textbf{(ii)}: Let $S \subseteq T \subseteq U$; notice that $S,T$ are not necessarily
    in the independence system $\mathcal I$. Then $|S \cap V| \le |T \cap V|$, so
    the function $f$ is monotone. 
    
    Next, let $x \in U \setminus T$. If $x \in V$, then 
    $$f( S \cup \{ x \} ) - f( S ) = f( T \cup \{x\} ) - f(T) = 1.$$ 
    If $x \in D$, 
    $$f( S \cup \{ x \} ) - f( S ) = f( T \cup \{x\} ) - f(T) = 0.$$ 
    Hence, in all cases, $f( S \cup \{ x \} ) - f( S ) \ge f( T \cup \{ x \} ) - f( T )$,
    so the function $f$ is submodular. 
\end{proof}
Next, we show that $\Phi$ is an approximation-preserving reduction. 
\begin{lemma} \label{prop1}
  By application of the reduction $\Phi$, any $\alpha$-approximation algorithm to $\texttt{MAXI}_1$ yields an $\alpha$-approximation
  to \isg.
\end{lemma}
\begin{proof}
  Let $G$ be an instance of \isg, and let $(\mathcal I, f) = \Phi(G)$. 
  Let $OPT_U = \max_{S \in \mathcal I} f(S)$. Since membership
  of a set $S \in \mathcal I$ requires that $S \cap V$ be edge-independent
  in $G$, we have
  that $OPT_U = OPT_G$, where $OPT_G$ is the maximum size
  of an edge-independent set of $G$. 
  Now suppose set $X \in \mathcal I$ satisfies $f( X ) \ge \alpha OPT_U$. 
  Then
  \begin{align*}
    \alpha OPT_G = \alpha OPT_U \le f(X) = |X \cap V|,
  \end{align*}
  and by definition of $\mathcal I$, $X \cap V$ is edge-independent
  in $G$. Therefore, any approximation algorithm for 
  $\texttt{MAXI}_1$ with ratio $\alpha$ yields an approximation
  algorithm for \isg with ratio $\alpha$ by the following method:
  given instance $G = (V,E)$ of \isg, transform to an instance
  $\Phi (G)$ of \maxiOne. Apply the $\alpha$-approximation to get
  set $S \in \mathcal I$ such that $f(S) \ge \alpha OPT_U$. Finally, project
  $S$ back to $V$ and return the edge-independent set $S \cap V$,
  which satisfies $|S \cap V| \ge \alpha OPT_G$.
\end{proof}
The next theorem follows  
from Lemma \ref{prop1} and the results
of \citet{Hastad1999}
on \isg: namely, for any $\epsilon > 0$,
there is no polynomial-time algorithm to approximate
\isg better than $|V|^{-1 + \epsilon}$ unless
NP = ZPP.
\begin{theorem} 
  For any $\epsi > 0$, there is no polynomial-time
  algorithm that achieves ratio better than
  $(2/|U|)^{1 - \epsi}$ on \maxiOne, where $U$ is the ground
  set of the instance of \maxiOne, unless NP = ZPP.
\end{theorem}
\begin{proof}
  For any $G = (V,E)$, the universe $U$ of $\Phi (G)$ 
  has $|U| = 2|V|$; by Lemma \ref{prop1} and 
  the result of \citet{Hastad1999}, the theorem follows.
\end{proof}
\section{The Greedy Ratio on \maxi, when $f$ is monotone} \label{sect:greedy-ratio}
When the function $f$ is monotone, we further
analyze the performance of the greedy algorithm 
 (Alg. \ref{alg:greedy}) on independence systems in this
section. When all maximal bases have equal size, 
we show that the greedy algorithm obtains
a ratio that matches our lower bound in the previous
section.

We begin with a performance ratio for the greedy
algorithm on an arbitrary independence system
in terms of the size $\beta$ of the largest independent
set. 
\begin{proposition} \label{greedy-ratio}
  Let $\mathcal I$ be an independence system,
  and let $\beta = \max_{S \in \mathcal I} |S|$. Let $G$
  be the solution returned by the greedy algorithm,
  and let $O \in \mathcal I$ be the optimal solution to \maxi.
  Then $f( G ) \ge f(O) / \beta$. 
\end{proposition}
\begin{proof}
  Let $U$ be the ground set of $\mathcal I$,
  and let $\alpha = \max_{x \in U : \{ x \} \in \mathcal I} f(x)$,
  and observe that $f(G) \ge \alpha$. Now let $S \in \mathcal I$;
  then by submodularity, $f(S) \le \alpha |S|$. 
  It follows that $f(G) \ge f(O) / \beta$.
\end{proof}
The next corollary, combined with the hardness result from the
previous section, shows that if the independence system has 
two disjoint bases, the greedy algorithm is the optimal approximation
on systems where bases have equal size.
\begin{corollary}
  Let $\mathcal I$ be a system where maximal bases have equal size, with at least two disjoint bases. 
  Then the greedy algorithm 
  is a $\left({2}/{|U|} \right)$-approximation algorithm to \maxiP on $\mathcal I$.
\end{corollary}
\begin{proof}
  Let $A,B \in \mathcal I$ be bases of $\mathcal I$, such that $A \cap B = \emptyset$.
  Since $\mathcal I$ is a $1$-system, for some $t$, $|A| = |B| = t$; hence $|U| = n \ge 2t$. 
  Hence, $\beta = \max_{S \in \mathcal I} |S| = t \le n / 2$, so the result follows
  from Prop. \ref{greedy-ratio}.
\end{proof}
% \begin{corollary}
%   The greedy algorithm is a $(2/|U|)$-approximation on \maxiOne.
% \end{corollary}

\section{The TripleGreedy Algorithm} \label{sec:triple}
In this section, the TripleGreedy (\tg, Algorithm \ref{alg:sg}) is presented.
The algorithm \tg is the first nearly linear-time algorithm to approximately maximize 
a submodular function $f$ with respect to a $p$-extendible system.

We start with an abstract subproblem required by \tg.
\begin{deff}[\maxuni] Given $f: 2^{U} \to \reals$ and independence system $\mathcal I$,
 determine $A \in \mathcal I$, 
such that for any $B \in \mathcal I$, $f( A \cup B ) \le f(A)$.
Even if no such $A$ exists, by an $\alpha$-approximation to \maxuni, it is meant an
algorithm that finds
$A \in \mathcal I$, such that for any $B \in \mathcal I$,
$\alpha f( A \cup B ) \le f(A)$.
\end{deff}
Notice that $A \cup B$ in the requirement of \maxuni may not be a member of the independence system.
%Also, observe that Lemma \ref{lemm:gupta-fixed} shows that the standard greedy algorithm is a 
%$1/(p + 1)$-approximation to \maxuni when $\mathcal I$ is a $p$-extendible system. 

The \tg algorithm employs two subroutines, one to approximate the \maxuni problem
and one for the unconstrained maximization
problem; the unconstrained maximization problem is to determine 
$\argmax_{S \subseteq U} f(S)$.  Since a total of three calls to these subroutines are
required, and since variants of greedy algorithms may
be used for each subroutine, Alg. \ref{alg:sg} is termed TripleGreedy.
 First, \tg determines a set $A \in \pint$ approximating \maxuni
with the function $f$; second, \tg determines a set $B \in \pint$ is found approximating \maxuni with the restriction of $f$ to ${U} \setminus A$. Third, a set $A' \subseteq A$ is found, approximating the maximum value
of $f$ restricted to $A$. Finally, the set in $\{A,B,A'\}$ maximizing $f$ is returned.

We remark that \tg functions similarly to the algorithm for 
maximizing submodular functions with respect to
cardinality constraint developed in \citet{Gupta2010}; in place of \maxuni,
\citet{Gupta2010} simply uses the greedy algorithm. By abstracting out 
this subproblem, we see that 1) a performance ratio may be proved in
a much more general setting than cardinality constraint, namely for $p$-extendible systems,
and 2) the faster thresholding approach developed by \citet{Badanidiyuru2014} (\threshold) 
for monotone submodular maximization can be used for \maxuni, which results in nearly linear runtime. 

% \begin{algorithm}
%   \caption{\std$(f, \psystem )$: The StandardGreedy Algorithm} \label{alg:sg}
%    \begin{algorithmic}[1]
%      \STATE {\bfseries Input:} $f : 2^{U} \to \reals$, 
%      $\psystem :$ collection of $p$ matroids,
%      $\epsi > 0$. 
%      \STATE {\bfseries Output:} $A \subseteq U$, such that $A \in \pint$.
%      \STATE $A \gets \emptyset , B \gets 2^{U}$
%      \WHILE{$B \neq \emptyset$}
%      \STATE $a \gets \argmax_{x \in B} f_x( A )$ 
%      \IF{$A + a \in \pint$}
%      \STATE $A \gets A + a$
%      \ENDIF
%      \STATE $B \gets B \setminus a$ 
%      \ENDWHILE
%    \STATE \textbf{return} $A$
% \end{algorithmic}
% \end{algorithm}
\begin{algorithm}[t]
   \caption{\tg$(f, \psystem )$: The TripleGreedy Algorithm}
   \label{alg:sg}
   \begin{algorithmic}[1]
     \STATE {\bfseries Input:} $f : 2^{U} \to \reals$, $\psystem$: $p$-extendible system
     \STATE {\bfseries Output:} $C \subseteq U$, such that $C \in \pint$.
     \STATE $A \gets $ \maxuni$(f, \psystem)$
     \STATE $g \gets f|_{{U} \setminus A}$
     \STATE $B \gets $ \maxuni$(g, \psystem)$
     \STATE $A' \gets $ \texttt{UNCONSTRAINED-MAX}$(f|_{A})$
   \STATE \textbf{return} $C \gets \argmax \{ f(A'), f(A), f(B) \}$
\end{algorithmic}
\end{algorithm}
If $f$ is submodular, then the
approximation ratio of \tg depends on the ratios of the algorithms
used for \maxuni and \unc. 
\begin{theorem} \label{thm:dg}
  Let $f: 2^{U} \to \reals$ be submodular, let $\psystem$ be an
  independence system,
  and let $O = \argmax_{S \in \pint} f(S)$, and let $C = $\tg$(f, \psystem)$.
  Then 
  \[ f(C) \ge \left( \frac{\alpha \beta }{\alpha + 2 \beta } \right) f(O). \]
  where $\beta$ and $\alpha$ are the ratios 
  of the algorithms used for \texttt{UNCONSTRAINED-MAX},
  and \maxuni, respectively.
\end{theorem}
\begin{proof}
  Let $A,A',B,C$ have their values at termination of \tg$(f,\mathcal I)$. 
  Suppose a $\beta$-approximation algorithm is used for \texttt{UNCONSTRAINED-MAX}. Then any set
  $D \subseteq A$ satisfies $f(D) \le \beta^{-1} \ff{A'}$. 
  Suppose an $\alpha$-approximation algorithm is used for \maxuni ; so
  $f( O \cup A ) \le \alpha^{-1} f(A)$ and
  $f( (O \setminus A) \cup B ) \le \alpha^{-1} f(B)$.
  \begin{align*}
    f( O ) \le f( \emptyset ) + f( O ) &\le f( O \cap A ) + f ( O \setminus A ) \\
    &\le \beta^{-1} \ff{A'} + f( O \cup A ) + f( (O \setminus A) \cup B ) \\
    &\le \beta^{-1} \ff{A'} + \alpha^{-1}f( A ) + \alpha^{-1}f( B ) \\
    &\le \left( \beta^{-1} + 2\alpha^{-1} \right) f(C), 
  \end{align*}
  where the second and third inequalities follow from the submodularity of $f$
  and the fact that $f$ is non-negative and $A \cap B = \emptyset$.
\end{proof}
\begin{algorithm}[t]
  \caption{\threshold$(f, \psystem )$: The ThresholdGreedy Algorithm of \citet{Badanidiyuru2014}}\label{alg:thresh}
   \begin{algorithmic}[1]
     \STATE {\bfseries Input:} $f : 2^{U} \to \reals$, 
     $\psystem :$ $p$-extendible system,
     $\epsi > 0$. 
     \STATE {\bfseries Output:} $A \subseteq 2^{U}$, such that $A \in \pint$.
     \STATE $A \gets \emptyset$
     \STATE $M \gets \max_{x \in U} f(x)$
     \FOR{$(\tau \gets M; \tau \ge \stopGain ; \tau \gets (1 - \epsi ) \tau)$}\label{line:forA}
     \FOR{$x \in U$}
     \IF{$f_x(A) \ge \tau$}
     \IF{$A + x \in \pint$}
     \STATE $A \gets A + x$
     \ENDIF
     \ENDIF
     \ENDFOR
     \ENDFOR
     \STATE \textbf{return} $A$
\end{algorithmic}
\end{algorithm}
Next, we establish that \threshold approximates \maxuni on $p$-extendible systems;
the proof is provided in Appendix \ref{apx}.
\begin{lemma} When $\mathcal I$ is a $p$-extendible system, the 
  \threshold algorithm (Alg. \ref{alg:thresh}) of  \label{lemma:threshold}
  \citet{Badanidiyuru2014} is a $\left( \left( \frac{p}{1-\epsi }+1 + \epsi \right)^{-1} \right)$-approximation for \maxuni .
\end{lemma}
Finally, by Theorem \ref{thm:dg} and Lemma \ref{lemma:threshold} we have
the ratio $\approx 1/(4 + 2p)$ in nearly linear time on $p$-extendible systems.
\begin{corollary} Let $\epsi > 0$. \label{corollary:triple}
If the deterministic
  $(1/2 - \epsi)$ approximation of \citet{Buchbinder2018} is used for \texttt{UNCONSTRAINED-MAX},
  and \threshold of \citet{Badanidiyuru2014} is used for \maxuni with ratio $\alpha= \left( \frac{p}{1 - \epsi}+1 + \epsi \right)^{-1}$, the ratio of \tg is 
  $\left( \frac{2}{1 - 2 \epsi} + \frac{2p}{1 - \epsi} + 2 + 2\epsi \right)^{-1}$ with $O \left( \frac{n}{\epsi} \log \left( \frac{n}{\epsi} \right) \right)$ queries to $f$ and to the independence system.
%, and at most
%  $O \left( \frac{pn}{\epsi} \log \left( \frac{n}{\epsi} \right) \right)$ matroid independence queries.
\end{corollary}

\clearpage
\bibliographystyle{plainnatfixed}
\bibliography{mend}
\clearpage
\appendix
\section{Appendix}
\label{apx}
\begin{proof}[Proof of Lemma \ref{lemma:threshold}]
Let $A = \{a_0, \ldots, a_{k} \} \in \mathcal I$ be returned  by \threshold.
Let $O \in \psystem$, $O \neq \emptyset$. The set $O$ 
will be partitioned into at most $k$ subsets $Y_i$,
each of size at most $p$, as follows. Let $O_0 = O$, $A_0 = \emptyset$.
Suppose $O_i, A_i$ have been obtained, such that $A_i \subsetneq O_i$, which
is initially satisfied at $i = 0$. 
By the definition of $p$-extendible system,
there exists $Y_i \subseteq O_i \setminus A_i$, with $|Y_i| \le p$, such
that $O_i \setminus Y_i + a_i \in \mathcal I$.
Then let $O_{i+1} = O_i \setminus Y_i + a_i$ and let $A_{i +1} = A_i + a_i$; 
clearly $A_{i+1} \subseteq O_{i+1}$.
If $A_{i+1} = O_{i+1}$, stop; otherwise, continue inductively until $i = k$.
Let $j \le k$ be the index at which this procedure terminates.
If $A_j \subsetneq O_j$, let $R_j = O_j \setminus A_j$
and redefine $O_i = O_i \setminus R_j$ for all $0 \le i \le j$.

\begin{claim} \label{claim:obv}
  For each $i$, $0 \le i \le j$, $A_i \cup \{ y \} \in \mathcal I$ for all 
  $y \in Y_i$. 
\end{claim}
\begin{proof}
  Since $A_i \cup \{ y \} \subseteq O_i$, and $O_i \in \mathcal I$, the claim follows by definition
  of independence system.
\end{proof}
\begin{claim} \label{claim:diff-thresh}
  $$f( O \cup A ) - f( O_0 \cup A ) \le \epsi M.$$
\end{claim}
\begin{proof}
  \begin{align*}
    f( O \cup A ) - f(O_0 \cup A) &= f( O_0 \cup R_j \cup A ) - f( O_0 \cup A ) \\
                                 &\le \sum_{r \in R_j} f( O_0 \cup A \cup \{ r \}) - f( O_0 \cup A ) \\
                                 &\le \sum_{r \in R_j} f( A \cup \{ r \} ) - f( A ) \le \epsi M,
  \end{align*}
  where the last inequality is by the stopping condition of \threshold and
  the fact that $A = A_j \subseteq O_j \cup R_j$, so $A \cup \{ r \} \in \mathcal I$ for
  all $r \in R_j$. The other inequalities follow from submodularity and the definition
  of $R_j, O_0$.
\end{proof}

Then
\begin{align*}
  f( O \cup A ) - f(A) &\le f(O_0 \cup A) - f(A) + \epsi M \\
  &= \sum_{i=0}^{j-1} f( O_i \cup A ) - f( O_{i+1} \cup A ) + \epsi M\\
  &= \sum_{i=0}^{j-1} f( O_{i + 1} \cup A \cup Y_i ) - f( O_{i+1} \cup A ) + \epsi M\\
  &\le \sum_{i=0}^{j-1} \sum_{y \in Y_i} f( O_{i + 1} \cup A \cup \{y\} ) - f( O_{i+1} \cup A ) + \epsi M\\
  &\le \sum_{i=0}^{j-1} \sum_{y \in Y_i} f( A_i \cup \{ y \} ) - f( A_i ) + \epsi M\\
  &\le \sum_{i=0}^{j-1} \frac{p}{1 - \epsi} \cdot ( f( A_i \cup \{ a_i \} ) - f( A_i ) ) + \epsi M \le \frac{p}{1 - \epsi} f(A) + \epsi M,
\end{align*}
where the first inequality is by Claim \ref{claim:diff-thresh}, 
the first two equalities are by telescoping and the definition of $O_i, Y_i$,
the second and third inequalities are by submodularity. The fourth inequality holds
by the following argument: when $a_i$ was added to $A_i$, it holds that the threshold
$\tau$ has its initial value $M$, in which case $f(y) \le M$ for any $y \in Y_i$, or
all $y \in Y_i$ were not added during the previous threshold $\tau / (1 - \epsi)$.
Hence $f(A_i \cup \{ a_i \}) - f(A_i) \ge (1 - \epsi) (f(A_i \cup \{ y \}) - f(A_i))$ 
by submodularity. 
Since $M \le OPT$, the lemma follows.
\end{proof}
%\section{}

%%% Local Variables:
%%% mode: latex
%%% TeX-master: "./arxiv.tex"
%%% End:

\end{document}